\documentclass[a4paper,UKenglish]{article}

\usepackage{microtype}
\usepackage{fullpage}
\usepackage{graphics}
\graphicspath{{./figures/},{./figures/}}

\usepackage{authblk}

\usepackage{graphicx}
\usepackage{amsthm}
\usepackage{amsmath}
\usepackage{enumerate}
\usepackage{color}
\usepackage{xspace}
\usepackage{url}

\usepackage{amsfonts}\usepackage{amssymb}
\usepackage{thmtools}\usepackage{mathscinet}
\usepackage{thm-restate}
    \newtheorem{theorem}{Theorem}
    \newtheorem{corollary}[theorem]{Corollary}
    \newtheorem{lemma}[theorem]{Lemma}

    	\definecolor{darkgreen}{rgb}{0.01, 0.93, 0.29}
\definecolor{lightbrown}{rgb}{0.91, 0.4, 0.11}
\usepackage{framed}

\title{APX-Hardness and Approximation for the $k$-Burning Number Problem\thanks{The work of D. Mondal is partially supported by NSERC.}}

\author[1]{Debajyoti Mondal}
\author[2]{N. Parthiban}
\author[2]{V. Kavitha}
\author[3]{Indra Rajasingh}
 
\affil[1]{Department of Computer Science\\

University of Saskatchewan, Saskatoon, Canada\\

  \texttt{dmondal@cs.usask.ca}} 
\affil[2]{Department of Computer Science and Engineering\\

 SRM Institute of Science and Technology, Chennai, India\\

  \texttt{parthiban24589@gmail.com, kavitha.psk@gmail.com}}
\affil[3]{School of Advanced Sciences\\

Vellore Institute of Science and Technology, Chennai, India\\

  \texttt{indrarajasingh@yahoo.com}}

\usepackage[textsize=tiny]{todonotes}
\usepackage{verbatim}

\newcommand{\changed}[1]{{\color{black} #1}}

\newcommand{\jyoti}[1]{{\color{black} #1}}

\begin{document}
\maketitle              
\begin{abstract}
Consider an information diffusion process on a graph $G$ that starts with $k>0$ burnt vertices, and at each subsequent step, burns the neighbors of the currently burnt vertices, as well as $k$ other unburnt vertices. The \emph{$k$-burning number} of $G$ is the minimum number of steps $b_k(G)$ such that all the vertices can be burned within $b_k(G)$ steps. Note that the last step may have smaller than $k$ unburnt vertices available, where all of them are burned. The $1$-burning number coincides with the well-known burning number problem, which was proposed to model the spread of social contagion. The generalization to $k$-burning number allows us to examine different worst-case contagion scenarios by varying the spread factor $k$.  

In this paper we prove that computing $k$-burning number is APX-hard, for any fixed constant $k$. We then give an $O((n+m)\log n)$-time 3-approximation algorithm for computing $k$-burning number, for any $k\ge 1$, where $n$ and $m$ are the number of vertices and edges, respectively. Finally, we show that even if the burning sources are given as an input, computing a burning sequence itself is an NP-hard problem.

\end{abstract}

\jyoti{

\section{Introduction}

We consider an information diffusion process that models a social contagion over time from a theoretical point of view. At each step, the contagion propagates from the infected people to their neighbors, as well as a few other people in the network become infected. The burning process, proposed by Bonato et al.~\cite{doi:10.1080/15427951.2015.1103339,DBLP:conf/waw/BonatoJR14}, provides a simple model for such a social contagion process. Specifically, the \emph{burning number} $b(G)$ of a graph $G$ is the minimum number of discrete time steps or rounds required to burn all the vertices in the graph based on the following rule. One vertex is burned in the first round. In each subsequent round $t$, the neighbors of the existing burnt vertices and a new unburnt vertex (if available) are burned. If a vertex is burned, then it remains burnt in all the subsequent rounds. Figure~\ref{fig:intro}(a) illustrates an example of the burning process. The vertices that are chosen to burn directly at each step, form the \emph{burning sequence}. 

\begin{figure}[pt]
    \centering
    \includegraphics[width=.65\textwidth]{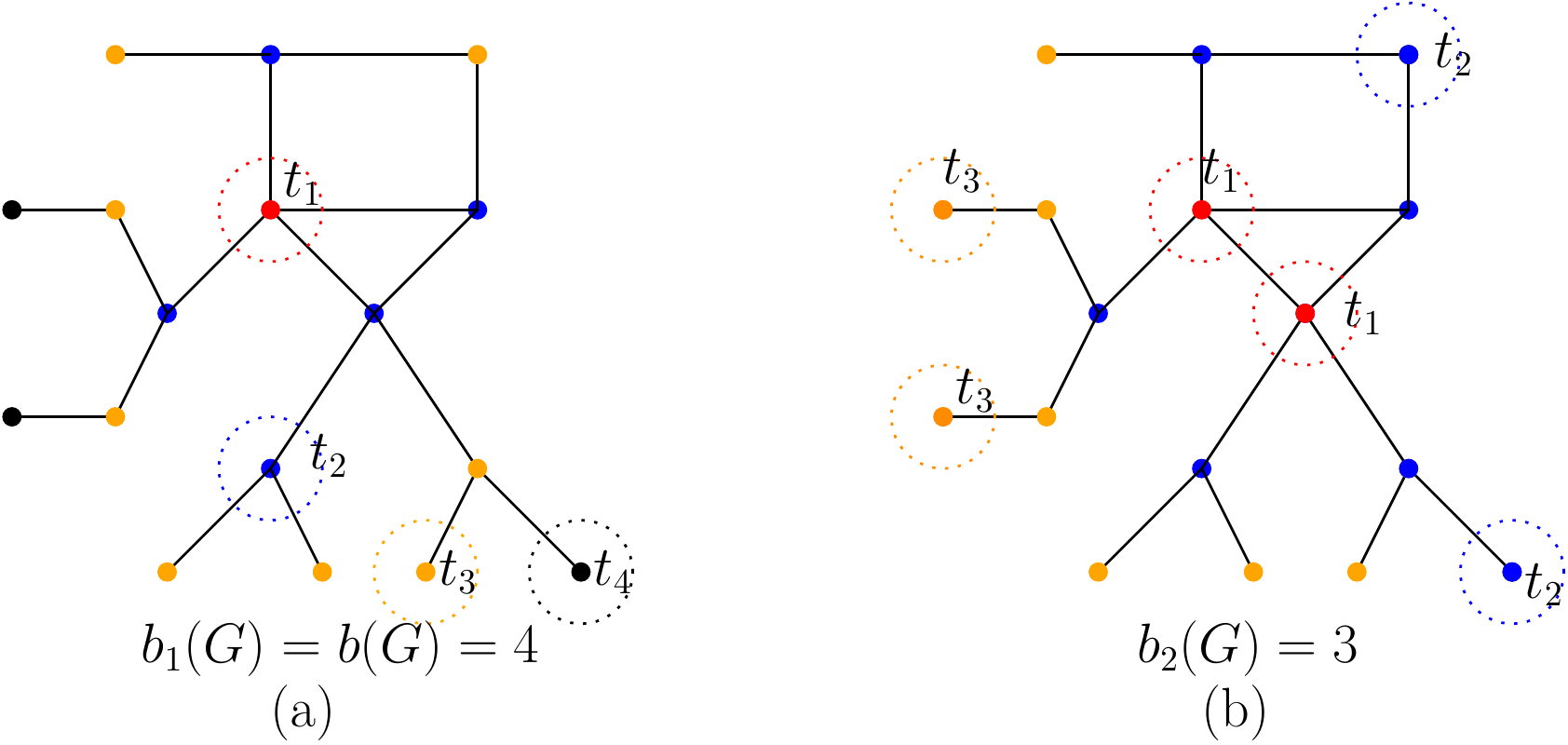}
    \caption{The process of burning a graph $G$. The unburnt vertices that have been chosen to burn at round $i$ (except for the neighbors of the previously burned vertices) are labelled with $t_i$.  (a) A 1-burning with  $4$ rounds, which is also the minimum possible number of rounds to burn all the vertices with 1-burning, i.e., $b(G)=4$. (b) A 2-burning with 3 rounds, which is the minimum possible, i.e., $b_2(G)=3$.}
    \label{fig:intro}
\end{figure}


In this paper we examine \emph{$k$-burning number} for a graph, which generalizes the burning number by allowing to directly burn $k$  unburnt vertices at each round; see Figure~\ref{fig:intro}(b). Throughout the paper, we use the notation  $b_k(G)$ to denote the $k$-burning number of a graph $G$. Note that in the case when $k=1$, the $1$-burning number $b_1(G)$ coincides with the original burning number $b(G)$. The burning process can be used to model a variety of applications, e.g, the selection of the vertices in social networks (e.g., LinkedIn or  Facebook) to fast spread information to the target audience with a  pipeline of steady new recruits. It may also be used in predictive models to examine the worst-case spread of disease. The generalization of a burning process to $k$-burning allows us to use $k$ as a model parameter, i.e., one can choose a cost-effective value for $k$ to increase the probability of reaching the target audience. 

\textbf{Related Work: }
The problem of computing the burning number of a graph is NP-complete, even for simple graph classes such as trees with maximum degree three, and for forests of paths~\cite{Bessy2017}. A rich body of research examines upper and lower bounds for the burning number for various classes of graphs. Bonato et al.~\cite{Bonato2016,Bessy2017}  showed that for every connected graph $G$, $b(G) \leq 2\sqrt{n}-1$, where $n$ is the number of vertices, and conjectured that the upper bound can be improved to $\lceil\sqrt{n}\rceil$. While the conjecture is still open, Land and Lu~\cite{DBLP:conf/waw/LandL16} improved this bound to $\frac{\sqrt{6n}}{2}$. However, the $\lceil \sqrt{n}\rceil$ upper bound holds for spider graphs~\cite{DBLP:journals/tcs/BonatoL19} and for $p$-caterpillars with at least $2 \left \lceil \sqrt{n} \right \rceil - 1$ vertices of degree 1~\cite{Hiller2019}. 

The burning number problem has received considerable attention in recent years and nearly tight upper and lower bounds have been established for various well-known graph classes including generalized Petersen graphs~\cite{Sim2017}, $p$-caterpillars~\cite{Hiller2019}, graph products~\cite{Mitsche2018}, dense and tree-like graphs~\cite{Kamali2019} and theta graphs \cite{Liu2019}.  The NP-hardness of the burning number problem motivated researchers to study the parameterized complexity and approximation algorithms.  Kare and Reddy~\cite{kare2019parameterized} gave a fixed-parameter tractable algorithm to compute burning number parameterized by neighborhood diversity, and showed that for cographs and split graphs the burning number can be computed in polynomial time.  Bonato and Kamali~\cite{bonato2019approximation} showed that the burning number of a graph is approximable within a factor of 3 for general graphs and  2 for trees. They gave a polynomial-time approximation scheme (PTAS) forests of paths, and a polynomial-time algorithm when the number of paths is fixed. They also mentioned that `it might be possible that a PTAS exists for general graphs'.

A closely related model that relates to the burning process is the firefighter model~\cite{Firefighter}. In a firefighter problem,  a fire breaks out at a vertex, and at each subsequent step, the fire propagates to the undefended neighbors and the firefighter can defend a vertex from burning. The burnt and defended vertices remain so in the next steps. The problems seek to maximize the number of defended vertices. This problem does not have a constant factor approximation~\cite{DBLP:journals/algorithmica/AnshelevichCHS12}, which indicates that it is very different than the burning number problem. A variant of firefighter problem where $b\ge 2$ vertices can be defended at each step has   been shown not to be approximable within a constant factor~\cite{BAZGAN2013899}. There are many information diffusion models  and broadcast scheduling methods   in the literature~\cite{Information,MIDDENDORF1993281,Plane}, but the $k$-burning process seems to differ in the situation that at each step it allows $k$ new sources to appear anywhere in the graph, i.e., some new burn locations may not be in close proximity of the currently burnt vertices.
}

\jyoti{
\textbf{Our Contribution.} In this paper, we generalize the concept of burning number of a graph to $k$-burning number.  \changed{We first prove that computing burning number is APX-hard, settling the complexity question posed by Bonato and  Kamali~\cite{bonato2019approximation}. We  then show that the hardness result holds for $k$-burning number, for any fixed $k$.} 
 We \changed{prove} that $k$-burning number is $3$-approximable in polynomial time, for any $k\ge 1$, where a 3-approximation algorithm was known previously for the case when $k=1$~\cite{bonato2019approximation}. 
 Finally, we show that even if the burning sources are given as an input, computing a burning sequence itself is an NP-hard problem. 
}

\section{Preliminaries}
\jyoti{
In this section we introduce some notation and terminology.

Given a graph $G$, the \emph{$k$-burning process} on $G$ is a discrete-time process defined as follows: Initially, at time $t=0$, all the vertices are unburnt. At each time step $t \geq 1$, the neighbors of the previously burnt vertices are burnt.  In addition, $k$ new unburnt vertices are burned directly, which are called the \emph{burning sources} appearing at the $t$th step. If the number of available unburnt vertices is less than $k$, then all of them are burned. The burnt vertices remain in that state in the subsequent steps. The process ends when all vertices of $G$ are burned. The $k$-burning number of a graph $G$, denoted by $b_{k}(G)$, is the minimum number of rounds needed for the process to end. For $k=1$, we omit the subscript $k$ and simply use the notation $b(G)$.

The burning sources are chosen at every successive round form an ordered list, which is referred to as a \emph{$k$-burning sequence}.   A burning sequence corresponding to the minimum number of steps $b_{k}(G)$ is referred to as a \emph{minimum burning sequence}. We use the notation $L(G,k)$ to denote the length of a minimum $k$-burning sequence. 



Let $G=(V,E)$ be a graph with $n$ vertices and $m$ edges. A vertex cover is a set $S\subseteq V$ such that  at least one end-vertex of each edge belongs to $S$. A \emph{dominating set} of $G$ is a set $D\subseteq V$ such that every vertex in $G$ is either in $D$ or adjacent to a  vertex in $D$. An \emph{independent set} of $G$ is a set of vertices such that no two vertices are adjacent in $G$. A \emph{minimum vertex cover} (resp., minimum independent and dominating set) is a vertex cover (resp., independent and dominating set) with the minimum cardinality.  An independent set $Q$ is called \emph{maximal} if one cannot obtain a larger independent set by adding more vertices to $Q$, i.e., every vertex in $V\setminus Q$ is adjacent to a vertex in $Q$.  
 }



\jyoti{
\section{APX-Hardness}

In this section we show that computing burning number is an APX-hard problem, which settles the complexity question posed by Bonato and Kamali~\cite{bonato2019approximation}. We then show that the $k$-burning number problem is APX-hard for any $k\in O(1)$.

\subsection{APX-Hardness for Burning Number}
 \changed{A graph is called \emph{cubic} if all its vertices are of degree three.}  We will reduce the minimum vertex cover problem in cubic graphs, which is known to be APX-hard~\cite{ALIMONTI2000123}. Given an instance $G=(V,E)$ of the minimum vertex cover, we construct a graph $G'$ of the burning number problem. We then show that a polynomial-time approximation scheme (PTAS) for the burning number in $G'=(V',E')$ implies a PTAS for the minimum vertex cover problem, which contradicts that the minimum vertex cover problem is APX-hard.

\textbf{Construction of $G'$.} 
The graph $G'=(V',E')$ will contain vertices that correspond to the vertices and edges of $G$. Figures~\ref{fig:cons}(a)--(c)  illustrate an example for the construction of $G'$ from $G$. To keep the illustration simple, we used a maximum degree three graph instead of a cubic graph. 

To construct $V'$, we first make a set $S$ by taking a copy of the vertex set $V$. We refer to $S$ as the \emph{v-vertices} of $G'$. For every edge $(u,v) 	\in E$,  we include two vertices $uv$ and $vu$ in $V'$, which we refer to as the \emph{e-vertices} of $G'$. In addition, we add $(2n+3)$ isolated vertices  in $V'$, where $n=|V|$ is the number of vertices in $G$.

\begin{figure}[h]
    \centering
    \includegraphics[width=.9\textwidth]{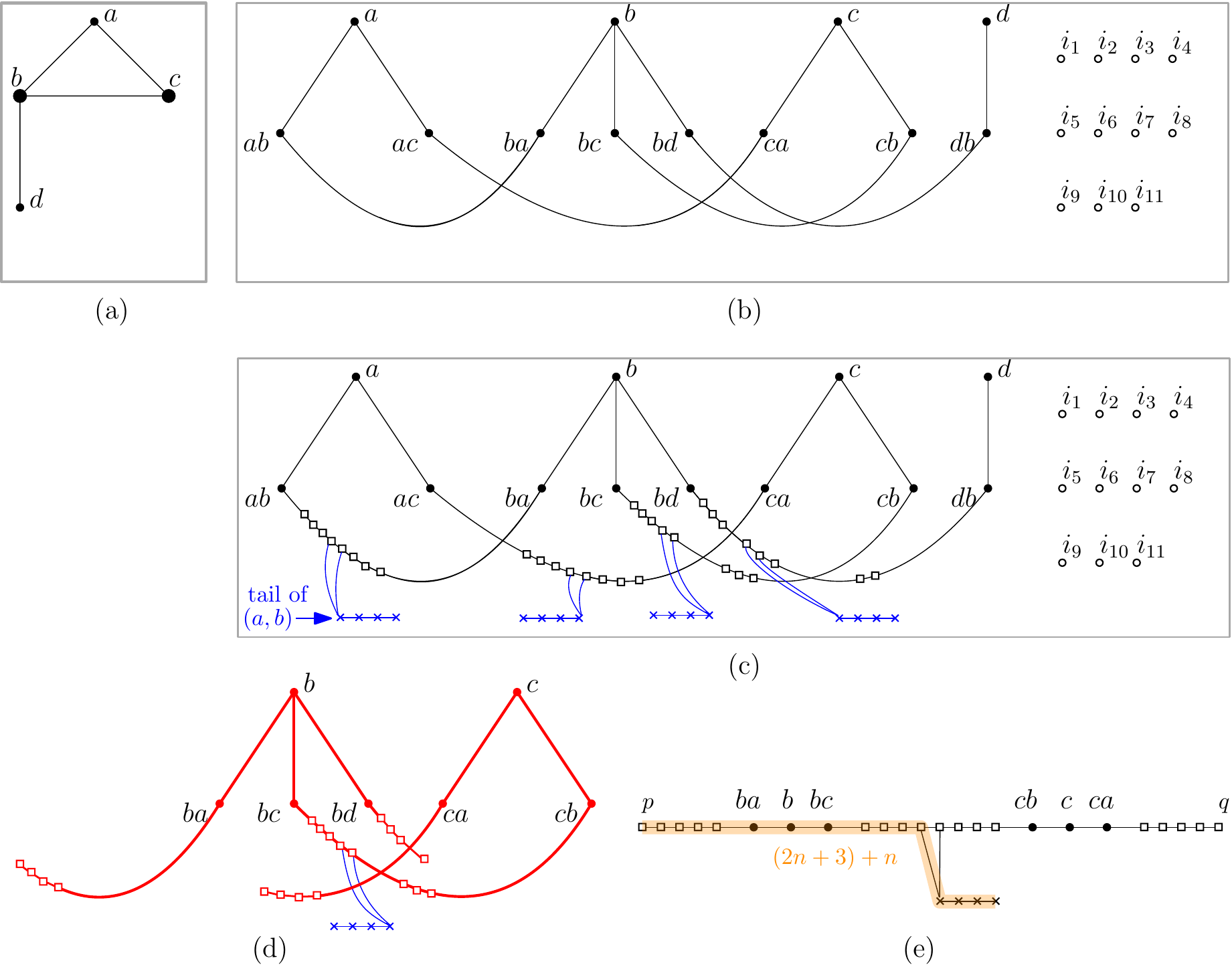}
    \caption{Illustration for the construction of $G'$ from $G$. To keep the illustration simple, here we use a maximum degree three graph instead of a cubic graph. (a) $G$, (b) construction idea, and (c) $G'$, where v- and e-vertices are shown in black disks, d-vertices are shown in squares, tail vertices are shown in cross, and isolated vertices are shown in unfilled circles. (d) $H_{bc}$. (e) Burning $H_{bc}$ with burning sources  $p$ and $q$, which are outside of $H_{bc}$.}
    \label{fig:cons}
\end{figure}

For every edge $(u,v)\in E$, we add three edges in $E'$:  $(u,uv),(v,vu)$ and $(uv,vu)$. Figure~\ref{fig:cons}(b) illustrates the resulting graph. We then divide the edge $(uv,vu)$ with $2n$ division vertices. We refer to these division vertices as the \emph{d-vertices} of $G'$. We also add an $n$-vertex path and connect one end with the median division vertices of the path $u,uv,\ldots,vu,v$. We refer to this $n$-vertex path as the tail of edge $(u,v)$.   Figure~\ref{fig:cons}(c) illustrates the resulting graph.

This completes the construction of $G'$. \changed{Note that the number of vertices and edges in $G'$ is $O(n)$, and it is straightforward to compute $G'$ from $G$ in $O(n)$ time.} 
}


\subsubsection{Reduction}
\jyoti{
 In the following, we show how to compute a burning sequence in $G'$ from a vertex cover in $G$, and vice versa.

\begin{lemma}\label{lem:fantastic}
If $G$ has a vertex cover of size at most $q$, then $G'$ has a burning sequence of length at most $(q+2n+3)$, and vice versa.    
\end{lemma}
\begin{proof}
We will use the idea of a neighborhood of a vertex. By a \emph{$r$-hop neighborhood} of a vertex $u$ in $G'$, we denote the vertices that are connected to $u$ \changed{by a path} of at most $r$ edges. 

\textbf{Vertex Cover to Burning Sequence:} Let $\mathcal{C}$ be a  vertex cover of $G$ of size at most $q$. In $G'$, we create a  burning sequence $S$ by choosing the v-vertices of $\mathcal{C}$ as the burning sources (in any order), followed by the burning of the $(2n+3)$ isolated vertices. Note that we need at most $q$ rounds to burn the v-vertices in $G'$ that correspond to the nodes in  $\mathcal{C}$, and in the subsequent  $(2n+3)$ rounds, we can burn the isolated vertices. 

We now show that all the vertices are burnt within $(q+2n+3)$ rounds. First observe that after $q$ rounds, all the v-vertices corresponding to $\mathcal{C}$ are burnt. Since $\mathcal{C}$ is a vertex cover,  all the v-vertices that do not belong to $\mathcal{C}$ are within $(2n+3)$-hop neighborhood from some vertex in $\mathcal{C}$. Therefore, all $v$-vertices will be burnt within the next  $(2n+3)$ rounds. Similarly, all the e-vertices, d-vertices and tail vertices are within $(2n+3)$-hop neighborhood from some vertex in $\mathcal{C}$, and thus they will be burnt within the next $(2n+3)$ rounds. Since the isolated vertices are chosen as the burning sources for the last $(2n+3)$ rounds, all the vertices of $G'$ will be burnt within $(q+2n+3)$ rounds.

\textbf{Burning Sequence to Vertex Cover:} We now show how to transform a given burning sequence $S$ of length $(q+2n+3)$ into a vertex cover $\mathcal{C}$ of $G$ such that $|C| \le q$. Let $S$  be the burning sources of the given burning sequence for $G'$. 

For every edge $(b,c)\in E$, we define $H_{bc}$ to be a subgraph of $G'$ induced by the $(n+1)$-hop neighborhood of $b$ and $c$, as well as the vertices on the path $b,bc,\ldots,cb,c$, and the vertices of the tail associated to $(b,c)$, e.g., see Figure~\ref{fig:cons}(d). 
For every $H_{bc}$ and for each burning source $w$ in it, we check whether $w$ is closer to $b$ than $c$. If  $b$ (resp., $c$) has a smaller shortest path distance to $w$, then we include $b$ (resp., $c$) into $\mathcal{C}$. We break ties arbitrarily.  


We now prove that $\mathcal{C}$ is a vertex cover of $G$. Suppose for a contradiction that there exists an edge $(b,c)\in E$, where neither $b$ nor $c$ belongs to $\mathcal{C}$. 
Then every burning source $s$ in $G'$ is closer to some $v$-vertex other than $b$ and $c$. In other words, $H_{bc}$ is empty of any burning source. Since $H_{bc}$ contains an induced path of $(n+1)+1+(2n+2)+1+(n+1) = (4n+6)$ vertices and a tail of $n$ vertices, burning all the vertices by placing burning sources outside $H_{bc}$ would take at least $(\frac{4n+6}{2}+n+1)$ steps, which is strictly larger than $(q+2n+3)$, e.g., see Figure~\ref{fig:cons}(d). 
Therefore, by construction of $\mathcal{C}$, at least one of $b$ and $c$ must belong to $\mathcal{C}$.  


It now suffices to show that the size of $\mathcal{C}$ is at most $q$. Since there are $(2n+3)$ isolated vertices in $G'$, they must correspond to $(2n+3)$ burning sources in the burning sequence. The remaining $q$ burning sources are distributed among the graphs $H_{uv}$. Therefore, $\mathcal{C}$ can have at most $q$ vertices. 
\end{proof}

}

\jyoti{
We now have the following theorem.
\begin{theorem} 
\label{t:apx}
The burning number problem is APX-hard.
\end{theorem}
\begin{proof}
Let $G$ be an instance of the vertex cover problem in a cubic graph, and let $G'$ be the corresponding instance of the burning number problem.  
By Lemma~\ref{lem:fantastic}, if $G$ has a vertex cover of size at most  $q$, then $G'$ has a burning sequence of length at most  $(q+2n+3)$, and vice versa. Let $C^*$ be a minimum vertex cover in $G$. Then $b(G') \le  |C^*|+2n+3$.

Let $\mathcal{A}$ be a  $(1+\varepsilon)$-approximation algorithm for computing the burning number, where $\varepsilon>0$. Then the burning number computed using $\mathcal{A}$ is at most $(1+\varepsilon) b(G')$. By Lemma~\ref{lem:fantastic}, we can use the solution obtained from  $\mathcal{A}$ to compute a vertex cover $\mathcal{C}$ of size at most $(1+\varepsilon) b(G') - 2n - 3$ in $G$. Therefore, $\frac{|C|}{|C^*|} = \frac{(1+\varepsilon) b(G') -2n- 3}{ |C^*|}
= \frac{b(G')+\varepsilon b(G') -2n- 3}{ |C^*|}
\le \frac{(|C^*|+2n+3) +\varepsilon b(G') -2n- 3}{ |C^*|}
= 1+\frac{\varepsilon b(G')}{ |C^*|}$.


Note that $G'$ has $n$ v-vertices, $(2n+3)$ isolated vertices, $2|E|$ e-vertices,   $n|E|$ tail vertices and $2n|E|$ d-vertices. Since $|E|\le 3n/2$,  the total number of vertices in $G'$ without the isolated vertices is upper bounded by $n+ 3n + n^2 + 3n^2 \le 4n^2+ 4n\le 5n^2$, for any $n> 4$. Since the burning number of a connected graph with $r$ vertices is bounded by $2\sqrt{r}$~\cite{Bessy2017}, the burning number of $G'$ is upper bounded by $(2n+3)+ 2\sqrt{5n^2} <  8n$, where the term $(2n+3)$ corresponds to the isolated vertices in $G'$. \changed{In other words, we can always burn the connected component first, and then the isolated vertices.} Furthermore, by Brooks' theorem~\cite{Brooks}, $|C^*|>n/3$.

We thus have $\frac{|C|}{|C^*|}\le 1+ \frac{\varepsilon b(G')}{|C^*|} \le 1+\frac{8n\varepsilon}{|C^*|}  \le 1+\frac{8n\varepsilon}{n/3} = 1+24\varepsilon$, which implies a polynomial-time approximation scheme for the minimum vertex cover problem. Hence the APX-hardness of burning number problem follows from the APX-hardness of minimum vertex cover.
\end{proof}

\paragraph{Hardness for Connected Graphs:} 
Note that in our reduction, $G'$ was disconnected. However, we can prove the hardness even for connected graphs as follows. Let $G$ be the input cubic graph, and let $v$ be a vertex in $G$. We create another graph $H$ by adding two vertices $w$ and $z$ in a path $v,w,z$. It is straightforward to see that the size of a minimum vertex cover of $H$ is exactly one plus the minimum vertex cover of $G$. We now carry out the transformation into a burning number instance $G'$ using $H$, but instead of using $(2n+3)$ isolated vertices, we connect them in a path $P=(w,Q,Q',i_1,Q',i_2,Q',\ldots, i_{2n+3},Q')$, where $Q$ is a sequence of $(q+2n+2)$ vertices, $Q'$ is a sequence of  $(2n+2)$ vertices, and $i_1,\ldots,i_{2n+3}$ are the vertices corresponding to the (previously) isolated vertices. Note that $P\setminus\{u,Q\}$ has $(2n+2)(2n+3) + (2n+3) = (2n+3)^2$ vertices. Since the burning number of a path of $r$ vertices is $\lceil \sqrt{r}\rceil $~\cite{Bessy2017}, any burning sequence will require $(2n+3)$ burning sources for  $P\setminus\{u,Q\}$. 

Note given a vertex cover $\mathcal{C}$ in $H$ of length $q$, if $w$ is not in $\mathcal{C}$, then $\mathcal{C}$ must contain $z$. Hence we can replace $z$ by $w$. Therefore, we can burn all the vertices within $(q+2n+3)$ rounds by burning $w$ first and then the other vertices of $\mathcal{C}$, and then the vertices of $P\setminus\{u,Q\}$ using the known algorithm for burning path~\cite{Bessy2017}. On the other hand, if a burning sequence of length $(q+2n+3)$ is provided, then $(2n+3)$ sources must be used to burn $P\setminus\{u,Q\}$. Since they are at least $(q+2n+3)$ distance apart from the vertices of $H$, at most $q$ burning sources are distributed in $H$, implying a vertex cover of size $q$. We thus have the following corollary.

\begin{corollary}
The burning number problem is APX-hard, even for connected graphs.
\end{corollary}
}
 
The generalization of the APX-hardness  proof for $k$-burning number is included in the Appendix. 


\section{Approximation Algorithms}

\jyoti{Bonato and Kamali~\cite{bonato2019approximation} gave an $O((n+m)\log n)$-time 3-approximation for burning number We
  leverage Hochbaum and Shmoys's~\cite{DBLP:journals/jacm/HochbaumS86} framework for designing the approximation algorithm and  give a generalized algorithm for computing $k$-burning number. For convenience, we first describe the 3-approximation algorithm in terms of Hochbaum and Shmoys's~\cite{DBLP:journals/jacm/HochbaumS86} framework.
 
}

\subsection{Approximating Burning Number} 
\jyoti{Here we show that for connected graphs, the burning number can be approximated within a factor of 3 in $O((n+m)\log n)$ time. 
 Let $G^{i}$ be the $i$th power of $G$, i.e., the graph obtained by taking a copy of $G$ and then connecting every pair of vertices with distance at most $i$ with an edge. We now have the following lemma.

\begin{lemma}
\label{lem:ub}
Let $G$ be a connected graph and assume that  $b(G)=t$. Then $G^t$ must have a dominating set of size at most $t$.
\end{lemma}
\begin{proof}
Since $b(G)=t$, all the vertices are burnt within $t$ rounds. Therefore, every vertex in $G$ must have a burning source within its $t$-hop neighborhood. Consequently, each vertex in $G^t$, which does not correspond to  a burning source in $G$, must be adjacent to at least one burning source. One can now choose the set of burning sources as the dominating set in $G^t$.
\end{proof}

\begin{figure}[pt]
    \centering
    \includegraphics[width=\textwidth]{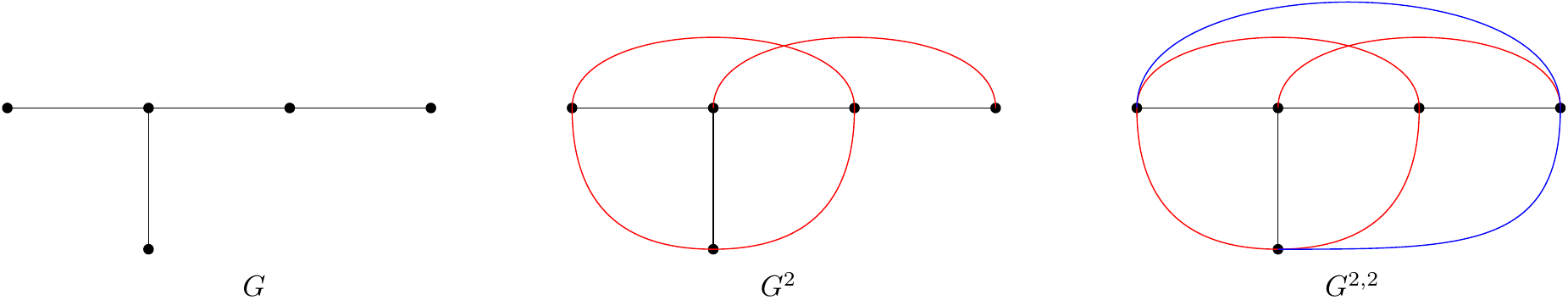}
    \caption{Illustration for the construction of $G^2$ and $G^{2,2}$ from $G$. 
}
    \label{fig:red}
\end{figure}

For convenience, we define another notation $G^{i,j}$, which is the $j$th power of  $G^i$. Although  $G^{i,j}$ coincides with  $G^{i+j}$, we explicitly write $i,j$. Let $M_{i,2}$ be a maximal independent set of  $G^{i,2}$.
 We now have the following lemma, which follows from the observation in~\cite{DBLP:journals/jacm/HochbaumS86} that the size of a minimum dominating set in $G$ is at least the size of a maximal independent set in $G^2$.  However, we give a proof for completeness.

\begin{lemma}
\label{lem:lb}
The size of a minimum dominating set in $G^{i}$ is at  least $|M_{i,2}|$. 
\end{lemma}
\begin{proof}
Let $Q$ be a minimum dominating set in $G^i$. It suffices to prove that for each vertex $v$ in $(M^{i,2}\setminus Q)$, there is a distinct vertex in $(Q \setminus M_{i,2})$ dominating $v$ (i.e., in this case, adjacent to $v$).

Let $\{p,q\}\subset (M_{i,2}\setminus Q)$ be two vertices in $G^i$, which are dominated by a vertex $w\in Q$ in $G^i$. Since $w$ is adjacent to $p,q$ in $G^{i,2}$ and $M_{i,2}$ is an independent set, we must have $w\in (Q\setminus M_{i,2})$. Since $w$ is adjacent to two both $p,q$ in $G^i$,  $p,q$ will be adjacent in $G^{i,2}$, which contradicts that they belong to the  independent set $M_{i,2}$. Therefore, each vertex in $(M^{i,2}\setminus Q)$, must be dominated by a distinct vertex in $(Q \setminus M_{i,2})$.
\end{proof}
}

\jyoti{
Assume that $b(G) = t$. By Lemma~\ref{lem:lb}, $G^t$ must have a dominating set of size at least $|M_{t,2}|$. By Lemma~\ref{lem:ub}, the size of a minimum dominating set $Q$ in  $G^t$ is upper bounded by $t$. We thus have the condition $|M_{t,2}|\le |Q|\le t$. 

\begin{corollary}
Let $G$ be a graph with burning number $t$ and let $M_{t,2}$ be a maximal independent set in $G^{t,2}$. Then $|M_{t,2}|\le t$.
\end{corollary}
Note that for any other positive integer $k<t$, the condition $|M_{k,2}|\le k$ is not guaranteed. We use this idea to approximate the burning number. We find the smallest index $j$, where $1\le j\le n$,  that satisfies  $|M_{j,2}| \leq j$ and prove that the burning number cannot be less than $j$.

\begin{lemma}
\label{lem:key}
Let $j'$ be a positive integer such that $j'<j$. Then $b(G)\not=j'$.  
\end{lemma}
\begin{proof}
Since $j$ is the smallest index satisfying $|M_{j,2}| \leq j$, for every other $M_{j',2}$, with $j' < j$ we have $|M_{j',2}| \ge j'+1$. 
 Suppose for a contradiction that $b(G)=j'$, then by Lemma~\ref{lem:ub}, $G^{j'}$ will have a dominating set of size at most $j'$. But by Lemma~\ref{lem:lb}, $G^{j'}$ has a minimum dominating set of size at least $|M_{j',2}| \ge j'+1$.   
\end{proof}

The following theorem shows how to compute a burning sequence in $G$ of length $3j$. Since $j$ is a lower bound on $b(G)$, this gives us a 3-approximation algorithm for the burning number problem.


\begin{theorem}
Given a connected graph $G$ with $n$ vertices and $m$ edges, one can compute a burning sequence of length at most $3b(G)$ in $O((n+m)\log n)$ time. 
\end{theorem}
\begin{proof}
Note that Lemma~\ref{lem:key} gives a lower bound for the burning number. We now compute an upper bound. We burn all the vertices of $M_{j,2}$ in any order. Since every maximal independent set is a dominating set, $M_{j,2}$ is a dominating set in $G^{j,2}$. Therefore, after the $j$th round of burning, every vertex of $G$ can be reached from some burning source by a path of at most $2j$ edges. Thus all the vertices will be burnt in $|M_{j,2}|+2j \leq 3j$ steps. Since $j$ is a lower bound on $b(G)$, we have $|M_{j,2}|+2j \leq 3j \leq 3b(G)$.

\begin{figure}[h]
    \centering
    \includegraphics[width=\textwidth]{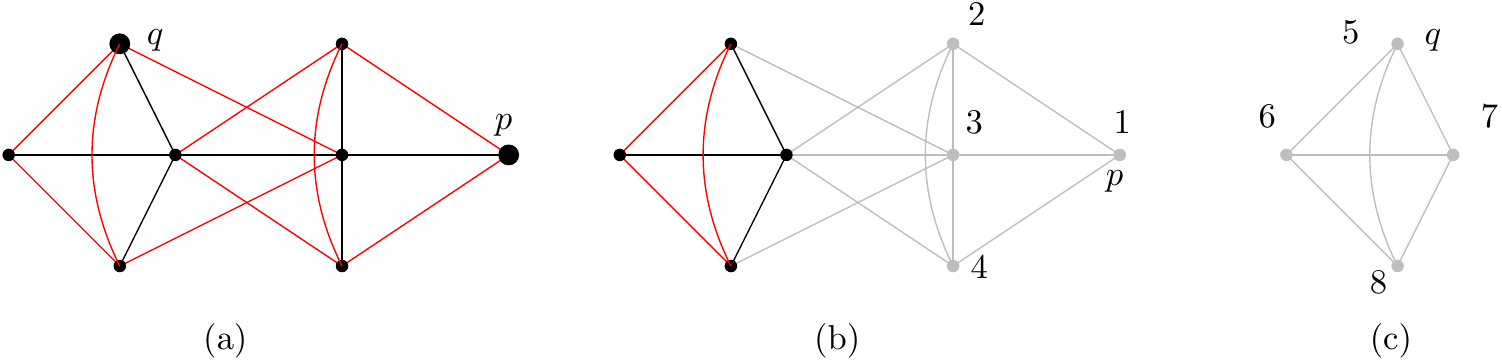}
    \caption{Illustration for computing $M_{r,2}$, when $r=1$. (a) $G^{1,2} = G^2$, where the edges of $G$ is shown in black, and a maximal independent set $M_{1,2} = \{p,q\}$. (b)--(c) Computation of $M_{1,2}$, where the numbers represent the order of vertex deletion.}
    \label{fig:maximal}
\end{figure}
It now suffices to show that the required $j$ can be computed in $O((n+m)\log n)$ time. Recall that $j$ is the smallest index satisfying $|M_{j,2}|\le j$. For any $j'>j$, we have $|M_{j',2}|\le |M_{j,2}| \le j < j'$. Therefore, we can perform a binary search to find $j$ in $O(\log n)$ steps. At each step of the binary search, we need to compute a maximal independent set $M_{r,2}$ in a graph $G^{r,2} = G^{2r}$, where $1\le r\le n$. To compute $M_{r,2}$, we repeatedly insert an arbitrary vertex $w$ of $G$ into $M_{r,2}$ and then delete $w$ along with its $r$-hop neighborhood in $G$ following a breadth-first order. Figure~\ref{fig:maximal} illustrates such a process. Since every edge is considered at most once, and the process takes $O(m+n)$ time. Hence the total time is $O((n+m)\log n)$.
\end{proof}

 
\subsection{Approximating $k$-Burning Number}

It is straightforward to generalize Lemma~\ref{lem:ub}  for $k$-burning number, i.e., if $b_k(G)=t$, then the size of  a minimum dominating set $Q$ in $G^t$ is at most $kt$. By Lemma 6, $G^t$ must have a dominating set of size at least $|M_{t,2}|$. Therefore, we have $|M_{t,2}|\le |Q| \le kt$.

Let $j$ be the smallest index such that  $|M_{j,2}|\le kj$. Then for any $j'<j$, we have $|M_{j',2}|>kj'$, i.e., every minimum dominating set in $G^{j'}$ must be of size larger than $kj'$. We thus have $b_k(G)\not = j'$. Therefore, $j$ is a lower bound on $b_k(G)$.

To compute the upper bound, we first burn the vertices of $M_{j,2}$. Since $|M_{j,2}|\le kj$, this requires at most $j$ steps. Therefore, after $j$ steps, every vertex has a burning source within its $2j$-hop neighborhood. Hence  all the vertices can be burnt within $3j \le 3 b_k(G)$ steps.

\begin{theorem}
The $k$-burning number of a graph can be approximated within a factor of 3 in polynomial time.
\end{theorem}
}

\section{Burning Scheduling is NP-Hard}
It is tempting to design heuristic algorithms that start with an arbitrary set of burning sources and then iteratively improve the solution based on some local modification of the set. However,  we show that even when a set of $k$ burning sources are given as an input, computing a burning sequence (i.e., burning scheduling) using those sources to burn all the vertices in $k$ rounds is NP-hard.

We   reduce the NP-hard problem 3-SAT~\cite{booknp}. Given an instance $I$ of 3-SAT with $m$ clauses and $n$ variables,  we   design a graph $G$ with $O(n^2+m)$ vertices and edges, and a set of $2n$ burning sources. We   prove that an ordering of the burning sources to burn all the vertices within $2n$ rounds can be used to compute an affirmative solution for the 3-SAT instance $I$, and vice versa (e.g., see Figure~\ref{fig:sch}). Due to space constraints, we include the details in the Appendix.

\begin{figure}[pt]
    \centering 
    \includegraphics[width=.75\textwidth]{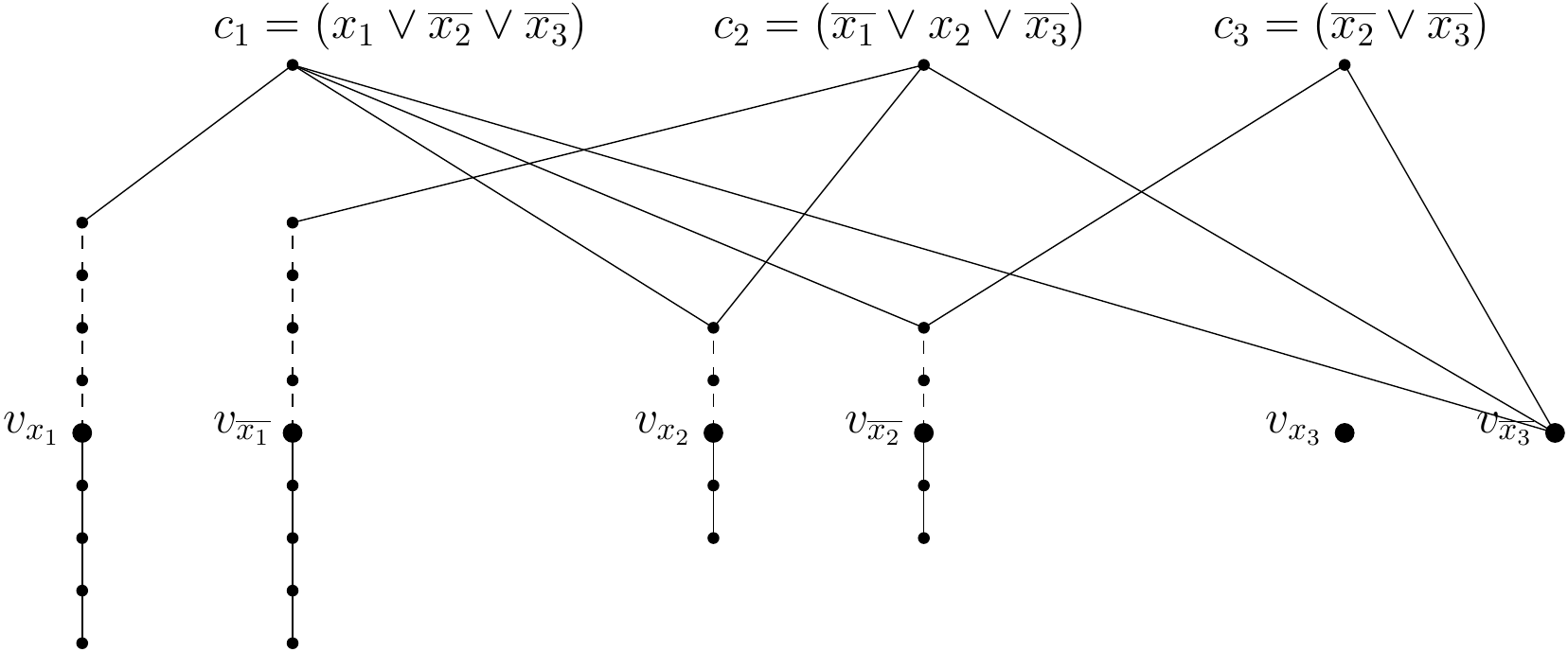}
    \caption{Illustration for the construction of $G$, where the given sources are shown in large disks. }
    \label{fig:sch}
\end{figure}
\section{Directions for Future Research}

A natural open problem is to find an improved approximation algorithm for $k$-burning number. One can also investigate whether existing approaches to compute burning number for various graph classes can be extended to obtain nearly tight bounds for their $k$-burning number. For example, the burning number of an $n$-vertex path is $\lceil\sqrt{n}\rceil$~\cite{Bonato2016}, which can be generalized to $\lceil\sqrt{n}/k\rceil$ for $k$-burning, as shown in the Appendix.

It would also be interesting to examine the edge burning number, where a new edge is burned at each step, as well as the neighboring unburnt edges of the currently burnt edges are burned. The goal is to burn all the edges instead of all the vertices. Edge burning number can be \changed{different than} the burning number, e.g.,  one can burn the vertices of every wheel graph in two rounds, but the edge burning number can be three. Given a graph, can we efficiently determine whether the burning number is equal to its edge burning number?

\bigskip
\noindent
\textbf{Acknowledgement.} We thank  Payam Khanteimouri, Mohammad Reza Kazemi, and Zahra Rezai Farokh for an insightful discussion that resulted into the addition of tail vertices while constructing  $G'$ in the proof of Lemma 1.
\bibliographystyle{abbrv}
\bibliography{Reference}

\newpage
\appendix
 

\section{APX-Hardness for $k$-Burning Number}
\jyoti{
For $k>1$, we use a similar reduction as we described for the burning number problem, except that we use a different number of division and isolated vertices. Given a decision problem that asks whether $G$ has a vertex cover of size $q$, we use $2nk$ division vertices for each edge, $n$ vertices for each tail, and $(k-1)q+k(2nk+3)$ isolated vertices to construct $G'$. We then show that the answer to the problem is affirmative if and only if $G'$ has a $k$-burning sequence of length $(q+2nk+3)$, and finally carry out the APX-hardness proof similar to  the proof of Theorem~\ref{t:apx}.  

\begin{theorem}
\label{thm:gen}
The $k$-burning number problem is APX-hard for every $k\in O(1)$.
\end{theorem}
\begin{proof}
We use a similar reduction as we described for the burning number problem, except that we use a different number of division and isolated vertices. Assume that $G$ has a vertex cover of size $q$. We use $2nk$ division vertices for each edge, $n$ vertices for each tail, and $(k-1)q+k(2nk+3)$ isolated vertices to construct $G'$. If $G$ has a vertex cover of size $q$, then the resulting graph $G'$ has a burning sequence of length at most $q+(2nk+3)$, as follows. We first burn $q$ vertices corresponding to the vertex cover along with $(k-1)q$ isolated vertices. Since every unburnt v-, e-, d- or tail vertex is within the $(2nk+3)$-hop neighborhood of some burnt vertex, they will be burnt in the next $(2nk+3)$ rounds by propagation. Furthermore, the isolated vertices will be directly burnt in the last $k(2nk+3)$ rounds.   

Assume now that $G'$ has a $k$-burning sequence $S$ of length $(q+2nk+3)$.  Burning the $(k-1)q+k(2nk+3)$ isolated vertices requires $(k-1)q+k(2nk+3)$ burning sources. Therefore, the remaining $q$ burning sources are distributed in the connected subgraph of $G'$. For any edge $(u,v)$,  we define $H_{uv}$ to be a subgraph of $G'$ induced by the $(nk+1)$-hop neighborhood of $u$ and $v$, as well as the vertices on the path $u,uv,\ldots,vu,v$, and the associated tail vertices. For every $H_{uv}$ and for each burning source $w$ in it, we check whether $w$ is closer to $u$ than $v$. If  $u$ (resp., $v$) has a smaller shortest path distance to $w$, then we include $u$ (resp., $v$) into $\mathcal{C}$. We break ties arbitrarily. 

We now prove that $\mathcal{C}$ is a vertex cover of $G$. Suppose for a contradiction that there exists an edge $(u,v)\in E$, where neither $u$ nor $v$ belongs to $\mathcal{C}$. Then every burning source $s$ in $G'$ is closer to some $v$-vertex other than $u$ and $v$. In other words, $H_{uv}$ is empty of any burning source. Since $H_{uv}$ contains an induced path of $(nk+1)+1+(2nk+2)+1+(nk+1) = (4nk+6)$ vertices and a tail of $n$ vertices, burning all the vertices by placing burning sources outside $H_{uv}$ would take at least $(\frac{4nk+6}{2}+n+1)$ steps, which is strictly larger than $(q+2nk+3)$. 
Therefore, by construction of $\mathcal{C}$, at least one of $u$ and $v$ must belong to $\mathcal{C}$.  

Therefore, the rest of the reduction can now be carried out following the argument presented in the proof of Theorem~\ref{t:apx}.  Note that  $\frac{|C|}{|C^*|} = \frac{(1+\varepsilon) b_k(G') -2nk- 3}{ |C^*|}
= \frac{b_k(G')+\varepsilon b_k(G') -2nk- 3}{ |C^*|}
\le \frac{(|C^*|+2nk+3) +\varepsilon b_k(G') -2nk- 3}{ |C^*|}
= 1+\frac{\varepsilon b_k(G')}{ |C^*|}$.  Since the number of vertices in $G'$ is upper bounded by $n+3n+n^2+ 3n^2k < 5n^2k$, we have $b_k(G') \le (k-1)q+k(2nk+3)+2\sqrt{5n^2k}< 8nk^2$, for sufficiently large $n$. Therefore,  $\frac{|C|}{|C^*|} 
\le 1+\frac{\varepsilon b_k(G')}{ |C^*|} 
\le 1+\frac{8nk^2\varepsilon}{ |C^*|}
=1+24k^2\varepsilon$, which contradicts the APX-hardness of minimum vertex cover.
\end{proof}
}

\section{Burning Scheduling is NP-Hard}

We will reduce the NP-hard problem 3-SAT~\cite{booknp}. Given an instance $I$ of 3-SAT with $m$ clauses and $n$ variables,  we will design a graph $G$ with $O(n^2+m)$ vertices and edges, and a set of $2n$ burning sources. We will prove that an ordering of the burning sources to burn all the vertices within $2n$ rounds can be used to compute an affirmative solution for the 3-SAT instance $I$, and vice versa.

\subsection{Construction of $G$} Let $x_1,\overline{x_1}, \ldots, x_{n},\overline{x_{n}}$ be the literals of $I$. For each literal $\ell$ we create a vertex $v_\ell$,  as shown in Figure~\ref{fig:sch} using large disks. We refer to these vertices as the \emph{literal vertices}. These vertices are the burning sources of $G$. For the $i$th positive literal, we create a path of $2(n-i)$ vertices and connect one of its ends with the literal vertex. In Figure~\ref{fig:sch}, these paths are drawn above the literal vertex with dashed edges. We will refer to these paths as the \emph{top paths} of $G$. Similarly, for the $i$th negative literal, we create a top path with $2(n-i)$ vertices and connect it to the literal vertices.  For each literal vertex in $G$, we now create a set of \emph{bottom paths} symmetrically, which are drawn below the literal vertices in Figure~\ref{fig:sch}.

For each clause $c_i$, $1\le i\le m$ in $I$, we create a vertex $v_{c_i}$. We refer to these vertices as the \emph{clause vertices}. For each literal in $c_i$, we add an edge between the literal vertex and $v_{c_i}$.  This completes the construction of $G$.

\subsection{Reduction} Let $L$ be the set of $2n$ literal vertices in $G$, which are the input burning sources. We now show that  an ordering of the burning sources that burns all the vertices within $2n$ rounds exists if and only if $I$ admits an affirmative solution. 

First assume that there exists an ordering of the burning sources that burns all the vertices within $2n$ rounds. For every source, if the round when it was burned is odd, then we set the corresponding literal to true. Otherwise, we set the literal to false, e.g., see Figure~\ref{fig:sch2}. We now prove that such an assignment will satisfy all the clauses of $I$.  

We first prove that for each index $j$ from $1$ to $n$, the literal vertices $v_{x_{j}}$ and  $v_{\overline{x_j}}$ must be burned within round  $2j$.  The reason is that both $v_{x_{j}}$ and  $v_{\overline{x_j}}$ has a bottom path with $2(n-j)$  vertices. If they are not burned within round $2j$, then we have at most $2n-(2j+1) = 2n-2j-1$ steps left which is smaller than the length of the bottom path. Therefore, $v_{x_1}$ and $v_{\overline{x_1}}$  are burnt within first two rounds, $v_{x_2}$ and $v_{\overline{x_2}}$  are burnt within the next two rounds, and so on. Therefore, our assignment based on odd and even label will consistently assign the truth values to the literals, e.g., if $x_1$ is set to true, then $\overline{x_1}$ is  set to false.

Suppose now for a contradiction that some clause $c$ is not satisfied. Then all its literals have are assigned even labels. By the construction of the top paths, each burning source with an even label can only have $2(n-j)$ rounds left to propagate the burning. Therefore, the propagation can only burn the top path, i.e., the burning does not reach the clause vertex $v_c$. This contradicts our initial assumption that the burning sources burn all the vertices within $2n$ rounds. 

\begin{figure}[h]
    \centering 
    \includegraphics[width=\textwidth]{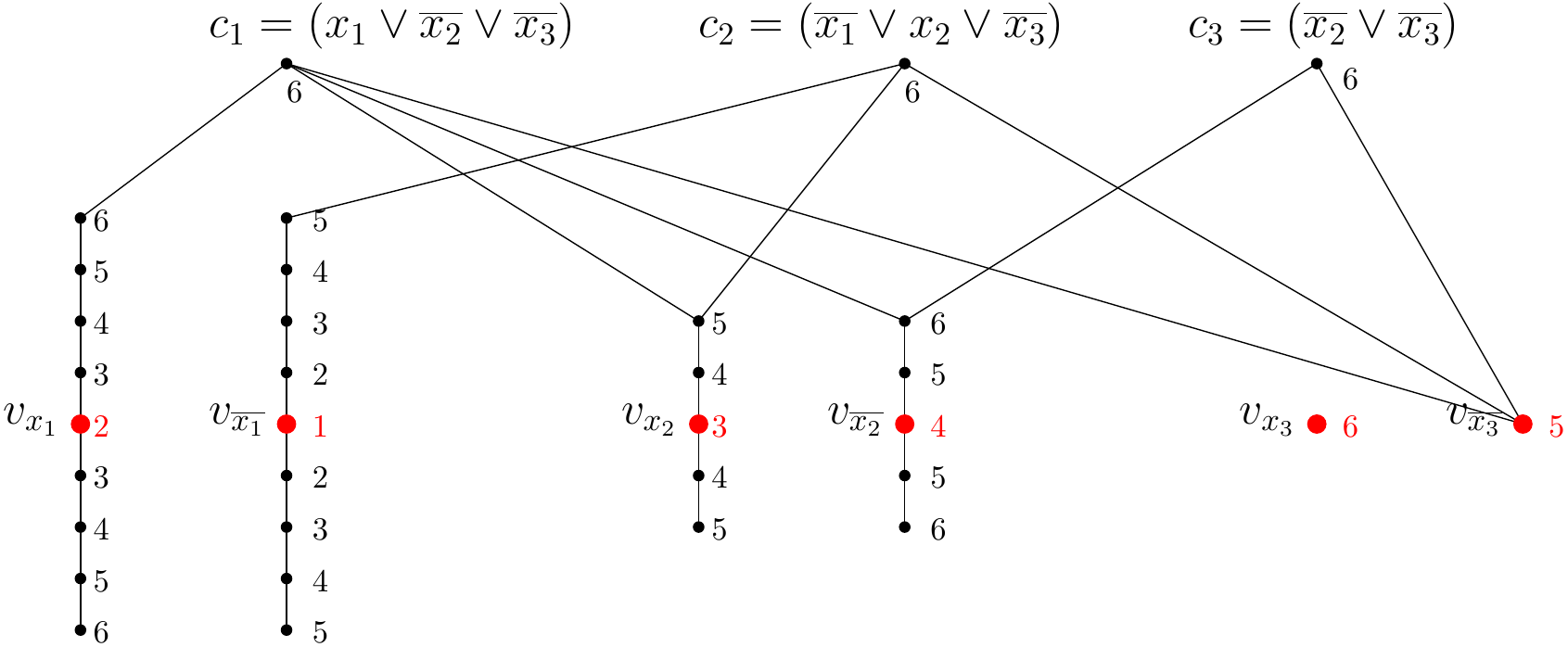}
    \caption{A burning sequence for the given sources that burns all the vertices. The literals corresponding to the sources with odd labels are set to true. In this case, $\overline{x_1}=1, x_2=1$, and $\overline{x_3}=1$. }
    \label{fig:sch2}
\end{figure}

Assume now that $I$ has an affirmative solution. For each index $j$ from $1$ to $n$, if ${x_{j}}$ is true, then we burn $v_{x_{j}}$ and  $v_{\overline{x_j}}$ at round  $(2j-1)$ and $2j$, respectively. Otherwise,  $\overline{x_{j}}$ is true, and we burn  $v_{\overline{x_j}}$ and $v_{x_{j}}$ at round  $(2j-1)$ and $2j$, respectively. Therefore, the number of rounds is $2n$. We now show that all the vertices of $G$ will be burnt. The vertices $v_{x_{j}}$ and  $v_{\overline{x_j}}$ are connected to top and bottom paths of length $2(n-j)$, and we have at least $2(n-j)$ steps left to burn these paths by propagation. Therefore, it suffices to show that the clause vertices are all burned.

Suppose now for a contradiction that some clause vertex $v_c$ is not burned. Since $c$ is satisfied, we can find a literal that is set to true.  Assume without loss of generality that the literal is a positive literal  $x_i$. Then the corresponding literal vertex $v_{x_i}$ is burned at round $(2j-1)$. Hence we will have $2n-(2j-1) = 2(n-j)+1$ steps left to propagate the burning beyond top path. Hence the vertex $v_c$ must be burnt within $2n$ rounds.   

The following theorem summarizes the results of this section.

\begin{theorem}
Given a graph $G$ and a set of $k$ burning sources. Finding a burning sequence for the given burning sources that burns all the vertices of $G$ within $k$ rounds is NP-hard.
\end{theorem}

\section{$k$-Burning Number for Paths}

\begin{theorem}
\label{l2}
If $G$ is an $n$-vertex path, then $b_{k}(G) =\lceil \sqrt{n/k} \rceil$. 
\end{theorem}
\begin{proof}
For a lower bound, observe that at the first round $k$ vertices are burned. In the next round, at most $2k$ new vertices are burned through propagation, as well as $k$ new vertices are burned. Let $S$ be the set of vertices that were burned by propagation and let $D$ be the vertices burned directly. Then in the $3$rd round, the set $S$ can lead to at most $S$ new burnt vertices by propagation, and the set $D$ may generate at most $2k$ new burnt vertices. Therefore, the number of vertices burned in this round is at most $|S|+2|D|+k = |S|+3k=5k$. In a general $m$th step, we burn at most $(2m-1)k$ new vertices.  Therefore, 
$k+3k+5k+\ldots+(2b_k(G)-1)k \ge n$, which implies $b_k(G)\ge\lceil \sqrt{n/k} \rceil$. 
 

For an upper bound, we can compute the partition of the path into $k$ subpaths, each with at most $\lceil n/k\rceil$ vertices. We can then apply the known algorithm for burning number~\cite{Bonato2016} to burn all the paths in parallel by a $k$-burning process in $\lceil \sqrt{\lceil n/k\rceil}\rceil = \lceil \sqrt{n/k} \rceil$ rounds. 
\end{proof}

\end{document}